\documentclass[12pt]{article}
\usepackage{eurosym}
\usepackage[all]{xy}
\usepackage{graphics,xcolor}
\usepackage{amsmath, amsfonts}
\usepackage{amssymb, graphicx}
\usepackage{float}
\usepackage{amsthm}
\usepackage[english]{babel}
\usepackage{hyperref}
\usepackage{epsfig}
\usepackage{amsmath}
\usepackage{enumerate}
\usepackage{setspace}
\usepackage{yhmath}
\usepackage{setspace}
\usepackage{natbib}
\usepackage{mathtools}

\newtheorem{theorem}{Theorem}

\newtheorem{corollary}{Corollary}

\newtheorem{example}{Example}

\newtheorem{lemma}{Lemma}

\begin{document}

\title{Solidarity to achieve stability\thanks{
Jorge Alcalde-Unzu and Oihane Gallo acknowledge the financial support from the Spanish Government through grant PID2021-127119NB-I00 funded by MCIN/AEI/10.13039/501100011033 and by ``ERDF A way of making Europe". Jorge Alcalde-Unzu acknowledges the financial support from Universidad P\'ublica de Navarra through grant PJUPNA2023-11403. Oihane Gallo acknowledges the Swiss National Science Foundation (SNSF) through Project 100018$\_$192583 as the main financial support. Oihane Gallo and Elena Inarra aknowledge the financial support from the Spanish Government through grants PID2019-107539GB-I00 funded by MCIN/AEI/10.13039/501100011033 and by ``ERDF A way of making Europe". Elena Inarra acknowledges financial support from the Basque Government through grant IT1697-22. Juan D. Moreno-Ternero acknowledges financial support from the Spanish Government through grant PID2020-115011GB-I00, funded by MCIN/AEI/10.13039/501100011033 and Junta de Andaluc\'{\i}a through grants P18-FR-2933 and A-SEJ-14-UGR20.}} 
\author{Jorge Alcalde-Unzu\thanks{
Universidad P\'{u}blica de Navarra, Spain; \href{mailto:jorge.alcalde@unavarra.es}{\tt jorge.alcalde@unavarra.es}} \and Oihane Gallo \thanks{University of Lausanne, Switzerland; \href{mailto:oihane.gallo@unil.ch}{\tt oihane.gallo@unil.ch}} \and Elena Inarra\thanks{University
of the Basque Country, Spain; \href{mailto:elena.inarra@ehu.eus}{\tt elena.inarra@ehu.eus}} \and Juan D. Moreno-Ternero\thanks{%
Universidad Pablo de Olavide, Spain; \href{mailto:jdmoreno@upo.es}{\tt jdmoreno@upo.es}}}
\maketitle

\begin{abstract} \medskip
\noindent Agents may form coalitions. Each coalition shares its endowment among its agents by applying a sharing rule. The sharing rule induces a coalition formation problem by assuming that agents rank coalitions according to the allocation they obtain in the corresponding sharing problem. We
characterize the sharing rules that induce a class of stable coalition formation problems as those that satisfy a natural axiom that formalizes the principle of solidarity. Thus, solidarity becomes a sufficient condition to achieve stability.

\medskip 

\medskip

\end{abstract}

\newpage

\section{Introduction}

\noindent According to the Merriam-Webster Dictionary, solidarity is defined as \textit{unity (as of a group or class) that produces or is based on community of interests, objectives, and standards}. It is a fundamental ethical principle that can be traced back to ancient philosophers such as Socrates and Aristotle and it is also connected to the slogan of the French revolution. It is also one of the six titles of the Charter of Fundamental Rights of the European Union. 

\medskip

\noindent The principle of solidarity has been often invoked in the axiomatic approach to economic design. 
It can be formally stated as follows: ``if the environment (e.g., resources, technology, population, or preferences) in which a group of people find themselves changes, and if no one in this group is responsible for the change, the welfare of all of them should be affected in the same direction: either they all end up at least as well off as they were initially, or they all end up at most as well off" \cite[]{thomson2021}. Axioms formalizing this idea have been crucial to characterize egalitarian allocation rules in diverse settings \cite[]{roemer1986equality, moulin1987egalitarian, moulin1989public, moreno2006impartiality, MARTINEZ2022740}. They have also been instrumental to characterize focal egalitarian rules in the theory of axiomatic bargaining and cooperative game theory \cite[]{kalai1975other, kalai1977proportional, thomson1980monotonicity, young1988distributive, chun1988monotonicity}. Thus, the egalitarian implications of the principle of solidarity have been well explored. Here, we focus on less explored implications, highlighting its role as a means to guarantee stability in contexts of coalition formation.

\medskip

\noindent Coalition formation has been the object of a large literature dealing with a plethora of social and economic problems such as cartel formation, lobbies, customs unions, conflict resolution, public goods provision, political party formation, etc. \cite[]{ray2007game, GRABISCH2012175, ray2015coalition}. A central concern in this literature has been stability, that is, immunity of a coalitional arrangement to ``blocking" \cite[]{perry1994noncooperative, seidmann1998theory,PULIDO2006860}. To be precise, a partition of agents into coalitions is (core) stable if there is no coalition in which each of its members strictly
prefers it to the coalition to which they belong in the partition.\footnote{Other stability concepts have been analyzed in the literature. See, for instance, \cite{karakaya2011hedonic} and \cite{gusev2021nash}.} Our focus here is stability for environments in which coalition members have an endowment to be shared among them. In these contexts, we can define for each agent an individual preference over the possible coalitions she can be part of, depending on the sharing rule used to distribute the endowment of each coalition: if the individual payoff an agent receives in one coalition is larger than her payoff in another, then this individual will prefer the former coalition to the latter. As preferences are constructed on the basis of a sharing rule, we can speak of ``the coalition formation problem induced by the sharing rule". And a natural question is whether there exist conditions on the sharing rule that guarantee the existence of stable partitions in its induced coalition formation problem. 
We show that such a condition is \textit{solidarity}, formalized by the following axiom: the arrival of new agents in a coalition, whether or not it is accompanied by a change in the endowment available to this coalition, should affect the welfare of all the incumbent agents in the same direction.

\medskip

\noindent Our solidarity axiom is reminiscent of the axiom of \textit{population monotonicity} \cite[]{thomson1983fair} but it is actually equivalent to the combination of two axioms that appear frequently in the literature: \textit{endowment monotonicity} and \textit{consistency}.\footnote{The solidarity axiom we consider was called \textit{population-and-resource monotonicity} in the context of rationing problems \cite[]{chun1999equivalence}. See also \cite{moreno2006impartiality}.} The former says that when a bad or good shock changes the endowment of a group, all its members should share in the calamity or the windfall. Thus, it has obvious solidarity underpinnings and it has long been used in axiomatic work \cite[]{moulin1988can,moulin1992application,MORENOTERNERO2021682,BERGANTINOS2022338}. The latter says that if a sub-group of agents secedes with the endowment allocated to their members then the rule allocates the remaining endowment to the standing agents in the same way. As such, it has been referred as a ``robustness", often ``coherence", or ``operational" principle \cite[]{thomson1998consistency, balinski2005just, moreno2012common, thomson2019divide, GUDMUNDSSON2023851}, although solidarity underpinnings have also been provided \cite[]{thomson2012axiomatics}. Alternative forms of consistency have indeed been suggested as axioms of stability in related contexts \cite[]{harsanyi1959, lensberg1987stability, lensberg1988stability}.

\medskip

\noindent Our main result actually states that a sharing rule satisfies \textit{solidarity} if and only if it induces \textit{non-circular coalition formation problems}. These problems preclude the existence of \emph{rings} (which arise when there is a cycle of coalitions with at least one agent in the intersection of any two consecutive coalitions who prefers the latter to the former, while the rest of the intersection-mates do not have the opposite preference) and satisfy the notion of \textit{weak pairwise alignment} (if one agent in the intersection of two coalitions ranks them in one way, no other agent in the intersection ranks them in the opposite way). As the non-circular coalition formation problems have a stable partition, our result implies that requiring solidarity in the sharing rule guarantees \emph{stability} in the induced coalition formation problem in the sense that these problems will have a non-empty core. 

\medskip

\noindent The closest research to our work is \cite{gallo2018rationing}, which combines coalition formation and rationing problems \cite[]{o1982problem}. Specifically, they assume that agents have ``claims" over the endowments associated to coalitions. The endowments are not sufficient for claims to be fully honored and, thus, are rationed by a rule. This induces agents' preferences over coalitions. Within the domain of continuous rules, they show that the properties of endowment monotonicity and consistency guarantee the existence of stable partitions in the induced coalition formation problems.\footnote{\cite{gallo2018rationing} wrongly state that these properties are not only sufficient but also necessary.} We generalize that result to any resource allocation situation (i.e., not necessarily rationing) and without considering the prerequisite of continuity.

\medskip

\noindent Another closely related paper is \cite{pycia2012stability}. We shall be more precise about the connection once we formally introduce our result later in the text. But we mention at least now that \cite{pycia2012stability} analyzes a general model of coalition formation (including cases in which not all coalitions are feasible) and shows that, under two regularity conditions, there is a stable partition if agents' preferences are generated by a sharing rule imposing Nash bargaining \citep{nash1950bargaining} over coalitional outputs.\footnote{From a different perspective, \cite{lensberg1987stability} obtains the same functional form.}  In contrast, our result characterizes rules inducing stability without resorting to the regularity conditions, but only referring to the case in which all coalitions are feasible.

\medskip

\noindent Apart from the applications to rationing and bargaining problems mentioned above, we show that our result can be applied to other settings. We first analyze surplus sharing problems, closely related to rationing problems, and then we develop an application based on \cite{dietzenbacher2023fair} who introduce the problem of prize allocations in competitions where agents are ranked.
Finally, we show that our result can also give some insights to resource allocation situations in which some coalitions are not feasible.

\medskip

\noindent The rest of the paper is organized as follows. In Section 2, we introduce the preliminaries of the model (sharing problems and coalition formation problems). In Section 3, we present our benchmark analysis and result. In Section 4, we present applications of our result to several problems such as bargaining, rationing, surplus sharing or ranking problems. Finally, we conclude in Section 5. The proof of the main result has been relegated to the Appendix.

\section{Preliminaries}

\subsection{Sharing problems}

Let $N$ be a finite set of agents. Consider a situation where a coalition of agents $C \subseteq N$ has an endowment $E\in \mathbb{R}_{+}$. A \textbf{sharing problem} is a pair $(C, E)$. Let $\boldsymbol{\mathcal{P}}$ denote the class of such problems.

\medskip

\noindent An \textbf{allocation} for $(C, E) \in \mathcal{P}$ is a vector $x=(x_{i})_{i\in C}\in \mathbb{R}_{+}^{|C|}$ that satisfies non-negativity, $0 \leq x_i$ for each $i\in C$, and efficiency, $\sum\nolimits_{i\in C} x_i=E$. A \textbf {(sharing) rule} is a function $F$ defined on $\mathcal{P}$ that associates with each $(C, E) \in \mathcal{P}$ an allocation $F(C,E)$ for $(C,E)$. The payoff of agent $i$ in problem $(C,E)$ under rule $F$ is denoted by $F_i(C,E)$. We denote by $\boldsymbol{\mathcal{F}}$ the set of all rules. 

\medskip

\noindent We now introduce several axioms for rules.

\medskip

\noindent The first axiom states that small changes in the endowment of the problem do not lead to large changes in the chosen allocation. 

\medskip

\noindent \textbf{Endowment continuity}: For each problem $(C, E) \in {\cal P}$ and each sequence of endowments $\{E_j\}_{j=1}^{\infty}$ with $E_j \rightarrow E$, 
$$F(C, E_j) \rightarrow F(C, E).$$

\noindent The second axiom states that if the endowment increases, then each agent receives at least as much as she initially did.

\medskip

\noindent \textbf{Endowment monotonicity}: For each pair of problems $(C,E), \linebreak (C^{\prime},E^{\prime}) \in \mathcal{P}$, with $C = C^{\prime}$ and $E<E^{\prime}$, and each $i\in C$, $$F_i(C, E)\leq F_i(C^{\prime},E^{\prime}).$$

\noindent The third axiom states that if some agents leave a coalition with the payoffs assigned to them by the rule, and the situation is reassessed at that point, then each remaining agent receives the same payoff as she initially did.

\medskip

\noindent \textbf{Consistency}: For each problem $(C, E) \in \mathcal{P}$, each $C' \subset C$, and each $i\in C'$,
$$F_i(C', \sum\limits_{i\in C'} F_i(C,E)) = F_i(C, E).$$

\noindent We finally introduce the axiom of solidarity: the possible arrival of new agents (with or without the endowment changing) does not affect the incumbent agents in opposite directions.

\medskip

\noindent \textbf{Solidarity}: For each pair of problems $(C, E), (C', E') \in \mathcal{P}$, with $C \subseteq C'$, and each pair $i, j \in C$, 
$$F_i (C, E) > F_i(C',E') \Rightarrow F_j(C, E) \geq F_j(C', E').$$

\noindent The next lemma states the relations between the previous axioms.

\begin{lemma}
\label{rm+c} The following statements hold:
\begin{itemize}
\item[(i)] If a rule satisfies \textit{endowment monotonicity}, then it also satisfies \textit{endowment continuity}.
\item[(ii)] A rule satisfies \textit{solidarity} if and only if it satisfies \textit{endowment monotonicity} and \textit{consistency}.
\end{itemize}
\end{lemma}

\begin{proof}
$(i)$ Let $F$ be a rule satisfying \textit{endowment monotonicity}.
Let $(C, E) \in {\cal P}$ and let $\{E_j\}_{j=1}^{\infty}$ be a sequence of endowments with $E_j \rightarrow E$. 
Then, for each $j$, $|E_j-E|=|\sum_{i\in C}F_i(C, E_{j})-\sum_{i\in C}F_i(C, E)|$. And, by \textit{endowment monotonicity}, for each $j$, 

$$|\sum_{i\in C}F_i(C, E_{j})-\sum_{i\in C}F_i(C, E)|=\sum_{i\in C}|F_i(C, E_{j})-F_i(C, E)|.$$

\noindent Thus, for each $i\in C$,  $|F_i(C, E_{j})-F_i(C, E)|\leq |E_j-E|$. As $|E_j-E|\to 0$, it follows that for each $i\in C$, $|F_i(C, E_{j})-F_i(C, E)|\to 0$. Equivalently, for each $i\in C$, $F_i(C, E_{j})$ converges to $F_i(C, E)$.

\medskip

\noindent $(ii)$ Assume first that $F$ is a rule that satisfies \textit{solidarity}. Then, it is straightforward to see that it also satisfies \textit{endowment monotonicity}. As for \textit{consistency}, let $(C, E) \in \mathcal{P}$ and $C' \subset C$. Then, by \textit{solidarity}, either for each $i\in C'$, $F_i (C, E) \geq F_i(C', \sum_{i \in C'} F_i(C, E))$, or for each $i\in C'$, $F_i (C, E) \le F_i(C', \sum_{i \in C'} F_i(C, E))$. Thus, for each $i\in C'$, $F_i (C, E) = F_i(C', \sum_{i \in C'} F_i(C, E))$, as desired. 

\medskip

\noindent Conversely, assume that $F$ is a rule that satisfies \textit{endowment monotonicity} and \textit{consistency}. Let $(C, E), (C', E') \in \mathcal{P}$ be a pair of problems with $C \subseteq C'$. By \textit{endowment monotonicity}, it follows that if $C = C'$ and $E \geq E'$, then for each $i \in C$, $F_i(C, E) \geq F_i(C', E')$, whereas if $C = C'$ and $E \leq E'$, then for each $i \in C$, $F_i(C, E) \leq F_i(C', E')$. Thus, either way, \textit{solidarity} holds. Assume next that $C \subset C'$. By \textit{consistency}, for each $i\in C$,
$$F_i (C, \sum_{i\in C} F_i(C', E')) = F_i(C', E').$$ 

\noindent By \textit{endowment monotonicity}, it follows that if $E \le \sum_{i\in C} F_i(C',E')$, then for each $i\in C$, $F_i(C, E) \le F_i( C, \sum_{i\in  C} F_i( C', E')) = F_i( C', E')$, whereas if $E\ge \sum_{i\in C}F_i( C', E')$, then for each $i\in C$, $F_i( C, E)\ge F_i( C, \sum_{i\in C} F_i( C', E'))= F_i( C', E')$. Thus, either way, \textit{solidarity} holds.
\end{proof}

\subsection{Coalition formation problems}

Consider a situation where each agent ranks the coalitions that she may belong to. Formally, let $N$ be a finite set of agents and $C\subseteq N$ denote a coalition. The collection of non-empty coalitions is denoted by $2^{N}$. For each agent $i\in N$, let $\succsim_{i}$ be a complete and transitive preference relation over coalitions containing $i$. Given $C,C^{\prime}\subseteq N$ such that $i\in C\cap C^{\prime}$, $C\succsim_{i}C^{\prime}$ means that agent $i$ finds coalition $C$ at least as desirable as coalition $C^{\prime}$. The binary relations $\succ_i$ and $\sim_i$ are defined as usual. A \textbf{(hedonic) coalition formation problem} is just a preference profile that consists of a list of preference relations, one for each $i\in N$, $\succsim=(\succsim_{i})_{i\in N}$. Let $\boldsymbol{\mathcal{D}}$ denote the class of such problems. 

\medskip

\noindent A partition is a set of non-empty coalitions whose union is $N$ and whose pairwise intersections are empty. Formally, a \textbf{partition} is a list $\pi=\{C_{1},\ldots, C_{m}\}$ such that \emph{(i)} for each $l=1, \ldots, m$, $C_l \neq \emptyset$, \emph{(ii)} $\bigcup_{l=1}^m C_l = N$, and \emph{(iii)} for each pair $l,l^{\prime} \in \{1, \ldots, m\}$, with $l \neq l^{\prime}$, $C_{l}\cap C_{l^{\prime}}=\emptyset$. Let $\boldsymbol{\Pi}$ denote the set of all partitions. For each $\pi \in \Pi$ and each $i\in N$, let $\boldsymbol{\pi(i)}$ denote the coalition in $\pi$ that contains agent $i$. A partition $\pi \in \Pi$ is \textbf{stable for $\boldsymbol{\succsim}$} if there is no coalition $T \subseteq N$ such that for each $i\in T$, $T \succ_{i} \pi(i)$. The set of all stable partitions for $\succsim$ is the \textbf{core} of $\boldsymbol{\succsim}$. The literature on coalition formation mostly focuses on identifying properties of the preference profiles that guarantee the existence of stable partitions.

\medskip

\noindent We now introduce several concepts and properties defined for coalition formation problems. First, a \textbf{\textit{ring}} is an ordered list of coalitions $(C_0, C_1, \ldots, C_{l-1})$, with $l>2$, such that for each subscript $k=0, 1, \ldots, l-1$ (modulo $l$) and each $j \in C_{k} \cap C_{k+1}$, $C_{k+1} \succsim_j C_{k}$, with at least one agent with strict preference in each intersection.\footnote{See, for instance, \cite{inal2015core} and \cite{pycia2012stability} for different definitions of rings, under the name of cycles.}  That is, in a ring there is at least one agent in the intersection of any two consecutive coalitions who prefers the latter to the former, while the rest of the intersection-mates are indifferent. It can be easily checked that the lack of rings guarantees the existence of a stable partition.

\medskip

\noindent The next property, originally introduced by \cite{pycia2012stability}, requires that all agents in the intersection of two coalitions rank them in the same way.

\medskip

\noindent \textbf{Pairwise alignment}: A coalition formation problem $\succsim \in {\cal D}$ is pairwise aligned if for each pair $C, C' \subseteq N$ and each pair $i,j\in C\cap C^{\prime}$, then [$C\succsim _{i}C^{\prime} \Leftrightarrow C\succsim _{j}C^{\prime}$].

\medskip

\noindent The \textbf{common ranking property}, introduced by \cite{farrell1988partnerships}, states that there is a common ranking of all coalitions that agrees with agents' preferences. Formally, a coalition formation problem satisfies the common ranking property if there is an ordering $\succsim$ over $2^{N}$ such that for each $i\in N$ and each pair $C,C'\subseteq N$ with $i\in C\cap C'$, $C\succsim_{i}C'\Leftrightarrow C\succsim C'$. Note that the common ranking property precludes the formation of rings. 
 In addition, when all coalitions are feasible, the common ranking property coincides with pairwise alignment (see Footnote 6 in \cite{pycia2012stability} for more details).  
\medskip

\noindent A weakening of the pairwise alignment property, introduced by \cite{gallo2018rationing}, requires that if one agent in the intersection of two coalitions ranks them in one way, no other agent in the intersection ranks them in the opposite way. 

\medskip

\noindent \textbf{Weak pairwise alignment}: A coalition formation problem $\succsim\in\mathcal{D}$ is weakly pairwise aligned if for each pair $C,C^{\prime} \subseteq N$ and each pair $i,j\in C\cap C^{\prime}$, then $[C\succ _{i}C^{\prime}\Rightarrow C\succsim_{j}C^{\prime}]$.

\medskip

\noindent Note that, unlike pairwise alignment, \textit{weak pairwise alignment} allows one agent to have a strict preference over two coalitions while any other agent in the intersection is indifferent between them. 

\medskip


\noindent The class of coalition formation problems that satisfy \textit{weak pairwise alignment} and do not have \textit{rings} is dubbed \emph{\textbf{non-circular coalition formation problems}} by \cite{gallo2018rationing}. This class includes the problems that satisfy the common ranking property (the proof is straightforward). It is also related to the class of problems that satisfy the \textbf{top-coalition property} \citep[see][]{banerjee2001core}. Formally, a coalition $C'\subseteq C$ is a top coalition of $C$ if for each $i\in C' $ and each $S\subseteq C$ with $i\in S$, we have $C'\succsim _{i}S$. A coalition formation problem satisfies the top-coalition property if each coalition $C\subseteq N$ has a top coalition. The non-circular coalition formation problems are included in the class of problems satisfying the top-coalition property \citep[see Theorem 1 in][]{gallo2018rationing}. Likewise, a weaker version of the top-coalition property guarantees the existence of a stable partition and, in consequence, so does the top-coalition property \citep[see Theorem 1 in][]{banerjee2001core}. To complete the relations between the properties, note that Examples 1 and 2 in \cite{gallo2018rationing} show that the top-coalition property does neither imply \textit{weak pairwise alignment} nor preclude the existence of \textit{rings}. Finally, Example 3 in \cite{gallo2018rationing} shows that \textit{weak pairwise alignment} does not guarantee the existence of a stable partition. All these relations among the properties are illustrated in Figure \ref{figure1}.

\begin{figure}[h]
\centering

\includegraphics[width=0.9\textwidth]{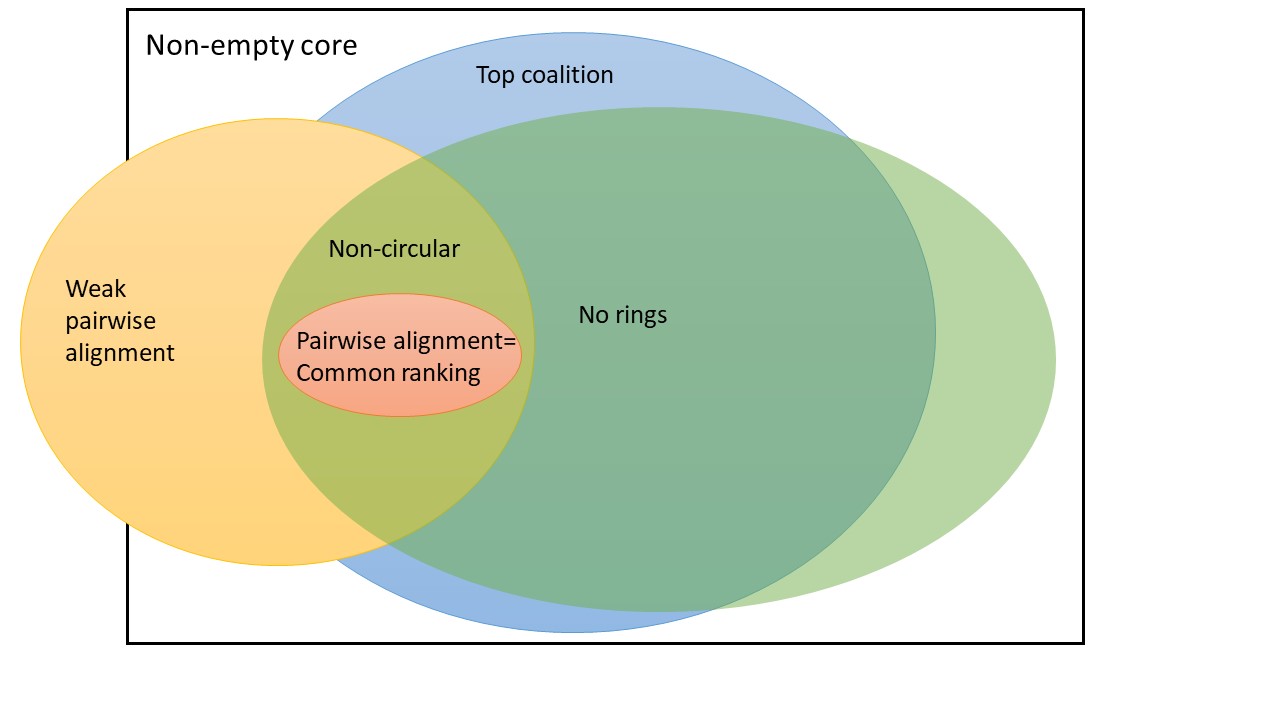}
\vspace{-1cm}\centering\caption{Relations among properties.}\label{figure1}
\end{figure}

\section{Benchmark analysis}\label{Section 3}

\noindent Given a set of sharing problems, one for each coalition, and a sharing rule, a coalition formation problem can be induced as follows: each agent computes her payoff in each problem with the sharing rule, and ranks coalitions accordingly. Formally, given a set of problems $\{(C, E_C)\}_{C \subseteq N, E_C \in \mathbb{R}_+}$, the \textbf{coalition formation problem induced by rule $\boldsymbol{F}\in \mathcal{F}$} is the list of preference relations $\succsim^{F} = (\succsim_{i}^{F})_{i\in N}$ defined as follows: for each $i\in N$, and each pair $C,C^{\prime} \subseteq N$ such that $i\in C\cap C^{\prime}$, $C \succsim_{i}^{F} C^{\prime}$ if and only if $F_i(C, E_{C}) \geq F_i(C^{\prime}, E_{C^{\prime}})$. 

\medskip

\noindent Our main result characterizes all rules that induce non-circular coalition formation problems. They happen to be those that satisfy the solidarity axiom. The proof can be found in the Appendix.

\begin{theorem}
\label{stable} A sharing rule $F$ satisfies solidarity if and only if for any set of sharing problems $\{(C, E_C)\}_{C \subseteq N, E_C \in \mathbb{R}_+}$, the coalition formation problem induced by $F$, $\succsim^F$, is non-circular.
\end{theorem}

\noindent The next result follows from Theorem 1 and the relations among properties presented above (see Figure 1).

\begin{corollary}\label{stable1}
If a sharing rule $F$ satisfies \textit{solidarity}, then for any set of sharing problems $\{(C, E_C)\}_{C \subseteq N, E_C \in \mathbb{R}_+}$, the core of the coalition formation problem induced by $F$, $\succsim^F$, is non-empty.
\end{corollary}

\medskip

\noindent Corollary \ref{stable1} implies that, when the sharing rule satisfies solidarity, then stability is guaranteed. These results are illustrated in Figure \ref{figure2}.

\begin{figure}[h]
\centering

\includegraphics[width=0.7\textwidth]{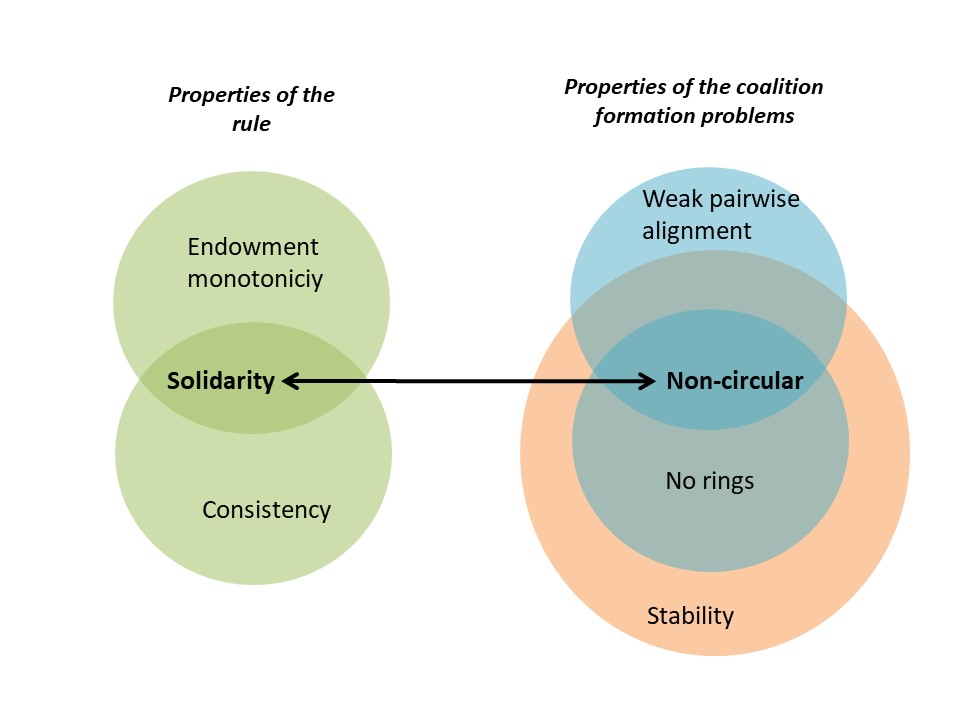}\caption{Sharing rules and their induced coalition formation problems.}\label{figure2}
\end{figure}

\medskip

\noindent Theorem \ref{stable} and Corollary \ref{stable1} are related to \cite{pycia2012stability}. On the one hand, note that our results apply to resource allocation situations generating hedonic coalition formation problems (i.e., all elements of $2^N$ are feasible coalitions) while his results also apply to other situations where not all coalitions are feasible, such as matching problems. Nevertheless, we show in Section 4 that our results can also give some insights for any resource allocation situation where only some coalitions are feasible.

\medskip

\noindent On the other hand, \cite{pycia2012stability} shows that pairwise alignment is a necessary and sufficient condition to guarantee a non-empty core in all the situations he analyzes (his Theorems 1 and 2). In contrast, by weakening pairwise alignment and adding absence of rings, we define the non-circularity property, which is weaker than pairwise alignment and it is sufficient to guarantee a non-empty core in our setting (where all coalitions are feasible). Consequently, this may lead to differences between the sharing rules that yield stability in the induced coalition formation problems in \cite{pycia2012stability} and those that do so in our setting. More precisely, \cite{pycia2012stability} characterizes the sharing rules obeying strict endowment monotonicity\footnote{Formally, for each pair of sharing problems $(C,E), (C,E^{\prime}) \in \mathcal{P}$, with $E<E^{\prime}$, and each $i\in N$, $F_i(C, E) < F_i(C,E^{\prime})$.} and non-satiability\footnote{Formally, for each $C \subseteq N$ and each $i \in C$, $\lim_{E \rightarrow \infty} F_i(C, E) = \infty$.} that guarantee that, for any of his resource allocation situations, all induced coalition formation problems have a non-empty core (his Corollary 1).\footnote{He also includes the axiom of \textit{endowment continuity}, but this axiom is implied by \textit{strict endowment monotonicity}, thanks to our Lemma \ref{rm+c} presented above.} However, we characterize the sharing rules that generate non-circular coalition formation problems by the axiom of solidarity (our Theorem \ref{stable}). Therefore, those sharing rules guarantee stability in the resource allocation situations where all coalitions are feasible (our Corollary \ref{stable1}). The above are not minor technical differences because, as we shall illustrate in the following example, some interesting sharing rules inducing stability in the situations where all coalitions are feasible are covered by our result, but not by Pycia's one.
\medskip

\begin{example}
\label{ex1}
\noindent Consider a sharing rule reflecting a situation in which an agent has priority over the rest of the agents. Moreover, she receives in each non-singleton coalition at most $k$ units of the endowment, while this bound does not apply for the other agents.
Then, $F$ allocates the first $k$ units of the endowment of any coalition to this agent (if she is in the coalition), while the remaining units of the endowment (if any are left) are shared equally among the other agents within the coalition. If the prioritized agent is not in the coalition, the rule simply imposes equal sharing of the endowment among all coalition members. 
  
\medskip

\noindent Formally, let $N = \{1,\ldots, n\}$ and consider the following sharing rule:

$$
F_i(C, E_C)=
\left\{
\begin{array}{cc}
\min\{E_C, k\} & \mbox{ if } i=1 \in C \, and \, C\neq\{1\},\\
\frac{E_C - \min\{E_C, k\}}{|C| - 1} & \mbox{ if } 1 , i \in C\, and\, i \neq 1 ,\\
\frac{E_C}{|C|} & \mbox{ otherwise. }
\\
\end{array}
\right.
$$

\noindent It can be checked that this rule satisfies \textit{solidarity}. Then, by our results, it also guarantees stability. We illustrate this by showing the existence of a stable partition for a particular specification. Let $N=\{1,2,3,4\}$, $k=10$, and the distribution of endowments 
be defined as follows:

\begin{table}[H]
\centering
\begin{tabular}{|c|c|c|c|c|c|c|c|c|}
\hline
$C$ & $\{12\}$ & $\{13\}$ & $\{14\}$ & $\{123\}$ & $\{124\}$ & $\{134\}$ &  $\{1234\}$ & otherwise
\\ \hline
$E_C$ & $10$ & $12$ & $14$ & $16$ & $20$ & $24$ &  $28$ & $0$ \\ \hline
\end{tabular}%
\end{table}


\noindent Then, $F$ yields the following allocations:

\smallskip

\small\begin{table}[h!]
	\centering
		\begin{tabular}{|c|c|c|c|c|c|c|c|c|}
			\hline
			$C$ &  $\{12\}$&  $\{13\}$ & $\{14\}$ & $\{123\}$ & $\{124\}$ & $\{134\}$& $\{1234\}$ & otherwise \\ 
			\hline
		$F(C, E_C)$ &$(10,0)$	&$(10,2)$ & $(10,4)$ & $(10,3,3)$ & $(10,5,5)$ & $(10,7,7)$ & $(10,6,6,6)$ & $(0)_{i\in C}$\\ \hline
		\end{tabular}
	
\end{table}


\noindent As a consequence, the coalition formation problem induced by $F$, $\succsim^{F}$, is the following:

\newpage

\begin{table}[h!]
\centering
\begin{tabular}{cccc}

 $\mathbf{1}$ & $\mathbf{2}$ & $\mathbf{3}$ & $\mathbf{4}$ \\ \hline
 $12\sim 13 \sim 14\sim$ & $1234$ & $134$ & $134$ \\ 
$\sim 123\sim 124\sim$ & $124$ & $1234$ & $1234$ \\ 
 $\sim 134\sim 1234$ & $123$ & $123$ & $124$\\ 
 $1$ &  $2 \sim 12 \sim 23 \sim$  & $13$ & $14$ \\ 
   & $\sim 24 \sim 234$ & $3\sim 23\sim$  & $4\sim 24\sim$ \\
   &   & $\sim34\sim 234$  & $\sim 34\sim234$ \\
\end{tabular}

\end{table}
\smallskip

\noindent Note that partitions $\{\{134\},\{2\}\}$ and $\{\{1234\}\}$ are stable. In this example $\succsim^{F}$ does not satisfy pairwise alignment; for instance, $\{124\} \succ_2^{F} \{123\}$, whereas $\{124\} \sim_1^{F} \{123\}$. However, it satisfies weak pairwise alignment and it has no rings, i.e., it is a non-circular problem. The reason why this rule is not included in Pycia's results is that it neither satisfies strict \textit{endowment monotonicity} nor \textit{non-satiability}.
\end{example}\medskip

\noindent The next example shows that, although the \textit{solidarity} axiom is sufficient to guarantee stability in the induced coalition formation problems, it is not necessary.
\medskip

\begin{example}
\label{ex2}
\noindent Consider a variant of Example \ref{ex1}, reflecting a situation in which agent $1$ has priority over the rest of the agents, but only when the grand coalition is formed. Moreover, she receives in the grand coalition at most $k$ units of the endowment, as before, while this bound does not apply
in other coalitions or for other agents. Then, $F$ allocates the first $k$ units of the endowment of the grand coalition to agent $1$, while the remaining units of the endowment (if any are left) are shared equally among the other agents. If the coalition is not the grand coalition, the rule simply imposes equal sharing of the endowment among all coalition members.\medskip

\noindent Formally, let $N = \{1,\ldots, n\}$ and consider the following sharing rule:

$$
F_i(C, E_C)=
\left\{
\begin{array}{cc}
\min\{E_C, k\} & \mbox{ if } i=1 \, and\, C=N,\\
\frac{E_C - \min\{E_C, k\}}{|C| - 1} & \mbox{ if }  i\neq 1 and\, C=N,\\
\frac{E_C}{|C|} & \mbox{ otherwise}.\\
\end{array}
\right.
$$

\noindent It can be checked that this rule does not satisfy \textit{solidarity}; in particular, it is not consistent. However, it never generates rings and, therefore, it induces coalition formation problems with a non-empty core. To see this, note that the rule imposes equal sharing for each coalition $C \neq N$. Thus, all agents agree on the ranking of all subcoalitions of $N$. Then, the common ranking property restricted to all subcoalitions of $N$ is satisfied.
Finally, it can be easily checked that, under this rule, bringing coalition $N$ into the picture does not generate any ring.


\medskip

\noindent We illustrate this rule for a particular specification. Let $N=\{1,2,3\}$, $k=6$, and the distribution of endowments be defined as follows:

\begin{table}[h!]
\centering
\begin{tabular}{|c|c|c|c|c|c|}
\hline
$C$ & $\{12\}$ & $\{13\}$ & $\{23\}$ & $\{123\}$ & otherwise
\\ \hline
$E_C$ & $10$ & $8$ & $6$ & $15$  & $0$\\ \hline
\end{tabular}%
\end{table}
\medskip

\noindent Then, $F$ yields the following allocations:
\smallskip

\begin{table}[h!]
\centering	
		\begin{tabular}{|c|c|c|c|c|c|}
			\hline
			$C$ &  $\{12\}$&  $\{13\}$ & $\{23\}$ & $\{123\}$ & otherwise \\ 
			\hline
		$F(C, E_C)$ &$(5,5)$	&$(4,4)$ & $(3,3)$ & $(6,4.5,4.5)$ & $(0)_{i\in C}$ \\ \hline
		\end{tabular}
	
\end{table}

\smallskip

\noindent As a consequence, the coalition formation problem induced by $F$, $\succsim^{F}$, is the following:

\smallskip

\begin{table}[H]
\centering
\begin{tabular}{ccc}

$\mathbf{1}$ & $\mathbf{2}$ & $\mathbf{3}$  \\ \hline
 $123$  & $12$ & $123$ \\ 
  $12$  & $123$ & $13$ \\ 
   $13$  & $23$ & $23$ \\ 
    $1$  & $2$ & $3$ \\ 
\end{tabular}%

\end{table}

\smallskip

\noindent Note that partitions $\{\{12\},\{3\}\}$ and $\{\{123\}\}$ are stable. 
However, $\succsim^{F}$ does not satisfy \textit{weak pairwise alignment} as $\{123\} \succ_1^{F} \{12\}$, whereas $\{12\} \succ_2^{F} \{123\}$.\medskip

\end{example}

\section{Applications and related work}

In many economic models, agents are characterized by different features (such as utility functions, claims, or ranking positions) that could be taken into account to distribute a given endowment among them. Our results show that, regardless of these characteristics, as long as the sharing rule satisfies solidarity, it will induce a coalition formation problem with a non-empty core. In this section, we first relate our results to existing results for bargaining problems \citep{pycia2012stability} and rationing problems \citep{gallo2018rationing}, and we then develop two novel applications for surplus sharing problems and ranking problems. Finally, we illustrate how our results can be applied to resource allocation situations with permissible coalitions.

\subsection*{Bargaining problems}

\noindent We first consider the case of coalition formation in bargaining problems introduced by \cite{pycia2012stability}. In this model, agents have utility functions which may be taken into account by the rule to get the final allocations. Formally, let $N = \{1, \ldots, n\}$ and let $u_{i} : \mathbb{R}_{+} \rightarrow \mathbb{R}_{+}$ denote agent $i$'s (non-decreasing) utility function for each $i\in N$. For each $C\subseteq N$, let $u_C=(u_{i})_{i\in C}$ be the profile of utility functions of coalition $C$ and $E_C\in\mathbb{R}_{+}$ be the endowment of coalition $C$. A \textbf{bargaining problem} is a triple $(C,E_C,u_C)$. An allocation for $(C,E_C,u_C)$ is a vector $x=(x_{i})_{i\in C} \in \mathbb{R}^{C}_{+}$ such that $\sum_{i\in C}x_{i}=E_C$. A \textbf{(sharing) rule} $F$ is a function that associates with each problem $(C,E_C,u_C)$ an allocation. Given a set of bargaining problems $\{(C,E_C,u_C)\}_{C\subseteq N, E_C \in \mathbb{R}_+}$, the \textbf{coalition formation problem induced by rule $\boldsymbol{F}$}, $\succsim^F$, is defined as in Section 3. 

\medskip

\noindent The two focal rules in this model are the so-called Nash bargaining solution, $N$, and Kalai-Smorodinsky bargaining solution, $KS$. The first one maximizes the product of agents' utilities \citep[see][]{nash1950bargaining}. The second one equalizes the relative gains - the gain of each player relative to its maximum possible gain - and maximizes this equal value \citep[see][]{kalai1975other}. Formally, for each bargaining problem $(C,E_C,u_C)$,
$$N(C,E_C,u_C)=\arg\max \prod\limits_{i\in C} u_i(x_i),$$
and $KS(C,E_C,u_C)$ is such that for each pair $i,j\in C$,
$$\frac{u_i(KS_i(C,E_C,u_C))}{u_i(E_C)}=\frac{u_j(KS_j(C,E_C,u_C))}{u_j(E_C)}.$$

\noindent \cite{pycia2012stability} shows that the Nash bargaining solution guarantees a non-empty core of the induced coalition formation problem. As the Nash bargaining solution satisfies solidarity, we know from our Theorem \ref{stable} that it induces non-circular coalition formation problems and, thus, stability is guaranteed. \cite{pycia2012stability} also shows that the core of the induced coalition formation problem from the Kalai-Smorodinsky bargaining solution can be empty for some coalitional endowments. As the Kalai-Smorodinsky bargaining solution fails to satisfy solidarity, we know from our Theorem \ref{stable} that it does not always induce non-circular coalition formation problems.

\subsection*{Rationing problems}

\noindent As we mentioned in the Introduction, our results generalize those obtained by \cite{gallo2018rationing} to any resource allocation situation (beyond rationing problems).
We analyze in this subsection how our results apply to that particular case of rationing problems. These problems pertain to situations where agents have claims over the endowment that cannot be fully honored and sharing rules take those claims into account to yield the allocations.\footnote{This model is renamed as generalized claims problems by \cite{gallo2022stable}.} Formally, let $N = \{1, \ldots, n\}$ and let $d_{i}\in\mathbb{R}_{+}$ denote agent $i$'s claim for each $i\in N$. For each $C\subseteq N$, let $d_C=(d_{i})_{i\in C}$ be the claims vector of coalition $C$ and $E_C\in\mathbb{R}_{+}$, with $E_C\leq\sum_{i\in C}d_i$, be the endowment of coalition $C$. A \textbf{rationing problem} is a triple $(C,E_C,d_C)$. An allocation for $(C,E_C,d_C)$ is a vector $x=(x_{i})_{i\in C}\in \mathbb{R}^{C}_{+}$ such that, for each $i\in C$, $x_i\leq d_i$ and $\sum_{i\in C}x_{i}=E_C$. A \textbf{(sharing) rule} $F$ is a function that associates with each problem $(C,E_C,d_C)$ an allocation. Given a set of claims problems $\{(C,E_C, d_C)\}_{C\subseteq N, 0 \leq E_C \leq \sum_{i \in C} d_i}$, the \textbf{coalition formation problem induced by rule $\boldsymbol{F}$},  $\succsim^F$, is defined as in Section 3.

\medskip

\noindent A focal family of sharing rules for rationing problems is the so-called family of \textit{parametric rules} \citep[]{young1987dividing}. The payoff of each agent given by a parametric rule is obtained by a function that only depends on her individual claim and a common parameter. Formally, a rule $S$ is \textit{parametric} if there exists a function $f:[a,b]\times \mathbb{R}_{+}\rightarrow \mathbb{R}_{+},$ where $[a,b]\subset \mathbb{R}\cup \{\pm \infty \},$ continuous and weakly monotonic in its first argument, such that:
\begin{itemize}
    \item[($i$)] $S_{i}(C, E_C, d_C)=f(\lambda, d_{i})$ for each problem $(C,E_C, d_C)$ and for some $\lambda \in \lbrack a,b]$, and
    \item[($ii$)] $f(a,x)=0$ and $f(b, x) = x$ for all $x\in \mathbb{R}_{+}$.
\end{itemize}

\noindent The family of parametric rules includes some well-known rules such as the proportional rule, $P$, the constrained equal-awards rule, $CEA$, and the constrained equal-losses rule, $CEL$ 
\citep[see][]{thomson2019divide}.\footnote{Formally, 
\[
f^{P}(\lambda ,d_i)= \lambda d_i \text{, for all }%
\lambda\in[0,1] \text{, }d_i\in \Bbb{R}_{+}\text{,} 
\]
\[
f^{CEA}(\lambda ,d_i)=\min \left\{ \lambda ,d_i\right\} \text{, for all }%
\lambda \text{, }d_i\in \Bbb{R}_{+}\text{,} 
\]
\[
f^{CEL}(\lambda ,d_i)=\max \left\{ 0,d_i-\frac{1}{\lambda }\right\} \text{%
, for all }\lambda\in[0,+\infty] \text{, }d_i\in \Bbb{R}_{+}\text{.} 
\]
}
Proposition 1 in \cite{gallo2018rationing} states that the parametric rules induce coalition formation problems with a non-empty core. As the parametric rules satisfy solidarity \citep[see][]{thomson2019divide}, our Theorem \ref{stable} guarantees that the induced coalition formation problems are non-circular and, therefore, by our Corollary \ref{stable1}, the core is non-empty.\footnote{\label{modif}In our benchmark analysis endowments are not constrained, whereas in this application they cannot be above the coalition's aggregate claim. Nevertheless, the proofs of Theorem \ref{stable} and Corollary \ref{stable} are also valid (with minor modifications) under that premise.}

\medskip

\noindent A focal non-parametric rule is the so-called \textit{random arrival rule}.\footnote{\cite{gallo2018rationing} refer to this rule as the Shapley value.} This rule gives each agent the average payoff after considering all possible arrival orderings and fully reimbursing each agent, in each order, until the endowment runs out \citep[see][]{o1982problem, thomson2019divide}. Formally, 
$$RA_i(C,E_C, d_C)=\frac{1}{|C|!}\sum_{\succ\in{\cal{O}}^{C}}\min\{d_{i},\max\{E-\sum_{j\in C,j\succ i}d_{j},0\}\},$$ where ${\cal{O}}^{C}$ denotes the set of strict orders in $C$. \cite{gallo2018rationing} show that the random arrival rule can generate a coalition formation problem with an empty core. As this rule does not satisfy \textit{solidarity}, we know from Theorem \ref{stable} that it does not always a induce non-circular coalition formation problem.

\subsection*{Surplus sharing problems}

\noindent We now consider the related case of coalition formation in surplus sharing problems. Surplus sharing problems \citep{moulin2002axiomatic} are opposite to rationing problems in the sense that endowments exceed the sum of claims. Formally, a \textbf{surplus sharing problem} for coalition $C$ is a triple $(C,E_C,d_C)$, with $E_C\geq\sum_{i\in C}d_i$. An allocation for $(C,E_C,d_C)$ is a vector $x=(x_{i})_{i\in C}\in \mathbb{R}^{C}_{+}$ such that, for each $i\in C$, $d_i\leq x_i$ and $\sum_{i\in C}x_{i}=E_C$. A \textbf{(sharing) rule} $F$ is a function that associates with each problem $(C, E_C, d_C)$ an allocation. Given a set of surplus sharing problems $\{(C,E_C, d_C)\}_{C\subseteq N, E_C \geq \sum_{i \in C} d_i}$, the \textbf{coalition formation problem induced by rule $\boldsymbol{F}$}, $\succsim^F$, is defined as in Section 3.\medskip

\noindent We define some focal surplus sharing rules. \noindent We start with the counterpart family of parametric rules in this setting \citep{moulin1987equal}. Formally, a (surplus sharing) rule $S$ is \textit{parametric} if there exists a function $f:[0,+\infty]\times \mathbb{R}_{+}\rightarrow \mathbb{R}_{+},$ continuous and weakly monotonic in both arguments, such that:
\begin{itemize}
    \item[($i$)] $S_{i}(C, E_C, d_C)=f(\lambda, d_{i})$ for each problem $(C,E_C, d_C)$ and for some $\lambda \in [0,+\infty]$, and
    \item[($ii$)] $f(0,x)=0$ and $f(+\infty, x) =+\infty$ for all $x\in \mathbb{R}_{+}$.
\end{itemize}
\noindent The family of parametric rules includes some well-known rules such as the proportional rule, $P$, the uniform gains rule, $UG$, and the equal surplus rule, $ES$ \citep[see][]{thomson2019divide}.\footnote{Formally,
\[
f^{P}(\lambda ,d_i)= \lambda d_i \text{, for all }%
\lambda\in[0,+\infty] \text{, }d_i\in \Bbb{R}_{+}\text{,} 
\]
\[
f^{UG}(\lambda ,d_i)=\max \{ \lambda ,d_i\},\text{ for all }%
\lambda\in[0,+\infty] \text{, }d_i\in \Bbb{R}_{+}\text{,} 
\]
\[
f^{ES}(\lambda ,d_i)=d_i+\lambda \text{, for all }\lambda\in[0,+\infty] \text{, }d_i\in \Bbb{R}_{+}\text{.} 
\]
}
All parametric rules for surplus sharing problems satisfy solidarity. Thus, our Theorem \ref{stable} guarantees that the induced coalition formation problems are non-circular and, as our Corollary \ref{stable1} states, that the core is not empty.\footnote{A similar caveat to the one made at Footnote \ref{modif} applies here.}

\medskip

\noindent However, there are also (non-parametric) rules for surplus sharing problems that do not satisfy solidarity. For instance, let us consider the following extension of the random arrival rule from rationing problems to surplus sharing problems. First, award all agents their claims as many times as the endowment allows. Then, assign the residual endowment (if it exists) sequentially according to an ordering of the agents. The \textit{extended random arrival for surplus sharing problems (ERA)} gives each agent the average payoff over all possible orderings. Formally, 
$$ERA_i(C,E_C, d_C)=kd_{i}+\frac{1}{|C|!}\sum_{\succ\in{\cal{O}}^{C}}\max\{E-k\sum_{j\in C}d_j-\sum_{j\succ i}d_{j},0\},$$ where 
$k=\left\lfloor{\frac{E_C}{\sum_{j\in C}d_j}}\right\rfloor$, 
and ${\cal{O}}^{C}$ denotes the set of strict orders in $C$.\footnote{Note that $ERA_i(C,E_C, d_C)=kd_{i}+RA_i(C,E_C-k\sum_{j\in C}d_j,d_C)$.}

\medskip

\noindent The following example shows that this rule does not guarantee stability and also illustrates that parametric rules do so.

\begin{example}
\noindent Let $N=\{1,2,3,4\}$ be such that $d=(1,3,3,10)$. Assume the following coalitional endowments:\smallskip

\begin{table}[th]
\centering
\begin{tabular}{|c|c|c|c|c|}
\hline
$C$ & $\{13\}$ & $\{23\}$ & $\{124\}$ & otherwise
\\ \hline
$E_C$ & $5$ & $8$ & $17$ & $\sum_{i\in C}d_i$ \\ \hline
\end{tabular}%
\end{table}

\noindent Then, the $ERA$ rule yields the following allocations:
\smallskip

\begin{table}[th!]
	\centering
		\begin{tabular}{|c|c|c|c|c|}
			\hline
			$C$ & $\{13\}$& $\{23\}$ & $\{124\}$ & otherwise \\ \hline
			$ERA(C,E_C,d_C)$ 
	& $(1.5, 3.5)$ & $(4,4)$ & $(1.33, 4.33, 11.33)$ & $(d_i)_{i\in C}$ \\ \hline
		\end{tabular}
	
\end{table}
\smallskip

\noindent As a consequence, the coalition formation problem induced by $ERA$, $\succsim^{ERA}$, is the following:

\begin{table}[th!]
\centering
\begin{tabular}{cccc}

 $\mathbf{1}$ & $\mathbf{2}$ & $\mathbf{3}$ & $\mathbf{4}$ \\ \hline
 $13$ & $124$ & $23$ & $124$ \\ 
 $124$ & $23$ & $13$ & $4\sim 14\sim 24\sim 34 \sim$ \\ 
$1 \sim 12 \sim 14 \sim$ &  $2\sim 12 \sim 24 \sim$ & $3\sim 34\sim 123 \sim$ & $\sim 134\sim 234\sim 1234$\\  
 $\sim 123 \sim 134 \sim 1234$ &  $\sim 123 \sim 234 \sim 1234$ & $\sim 134 \sim 234 \sim 1234$ & \\ \hline
\end{tabular}%

\end{table}

\noindent Observe that this problem is not a non-circular coalition formation problem. Although weak pairwise alignment is satisfied in this coalition formation problem, $(\{124\},\{13\},\{23\})$ is a ring and it can be easily checked that the core is empty.

\medskip

\noindent Consider now the \textit{uniform gains} rule (a parametric rule). This rule yields the following allocations in this problem:
\smallskip

\begin{table}[th!]

	\centering
		\begin{tabular}{|c|c|c|c|c|}
			\hline
			$C$ & $\{13\}$& $\{23\}$ & $\{124\}$ & otherwise\\ \hline
			$UG(C,E_C,d_C)$ 
	& $(2, 3)$ & $(4,4)$ & $(3.5, 3.5,10)$ & $(d_i)_{i\in C}$\\ \hline
		\end{tabular}
	
\end{table}
\smallskip

\noindent As a consequence, the coalition formation problem induced by $UG$, $\succsim^{UG}$, is the following:
\smallskip
\newpage

\begin{table}[th!]
\centering
\begin{tabular}{cccc}

 $\mathbf{1}$ & $\mathbf{2}$ & $\mathbf{3}$ & $\mathbf{4}$ \\ \hline
 $124$ & $23$ & $23$ & $4\sim 14\sim 24\sim 34 \sim$ \\ 
 $13$ & $124$ & $3 \sim 13 \sim 34 \sim 123 \sim$ & $\sim 124 \sim 134 \sim 234\sim 1234$ \\ 
 $1 \sim 12 \sim 14 \sim$ & $2\sim 12 \sim 24 \sim$ & $\sim 134 \sim 234 \sim 1234$ & \\  
$\sim 123 \sim 134 \sim 1234$ & $\sim 123 \sim 234 \sim 1234$ &  & \\ \hline
\end{tabular}

\end{table}

\noindent Observe that this coalition formation problem satisfies weak pairwise alignment and has no rings. Then, it is a non-circular coalition formation problem. In particular, partition $\{\{23\},\{1\},\{4\}\}$ is stable.
\end{example}

\subsection*{Ranking problems}
 
\cite{dietzenbacher2023fair} introduce the problem of prize allocation in competitions, in which a prize endowment has to be shared among the participants of a competition according to their ranking. With this idea in mind, we propose a model where agents are ranked and all coalitions can be formed. Then, the rule may take the ranking of the agents into account to derive the final individual payoffs. This model can be applied to any setting where agents can be ordered according to some characteristic (such as their expertise or past performance).

\medskip

\noindent Formally, let $N=\{1, \ldots, n\}$ be the set of agents. A \textbf{ranking} $\mathcal{R}$ is a bijection $\mathcal{R}:N \longrightarrow N$ that assigns to each agent a position, i.e., $\mathcal{R}(i)$ is the position of agent $i$ in the ranking. We say that agent $i\in N$ has a higher position in the ranking than agent $j\in N$ if $\mathcal{R}(i)<\mathcal{R}(j)$. For each $C \subseteq N$, let $E_C$ denote the endowment of coalition $C$ and $\mathcal{R}_C$ the projection of ranking $\mathcal{R}$ to $C$. That is, the position of agent $i$ in coalition $C$ is the number of agents, including herself, that have a higher position in that coalition. Formally, for each $i\in C$, $\mathcal{R}_C(i)=|\{j\in C: \mathcal{R}(i)\leq\mathcal{R}(j)\}|$. For each coalition $C \subseteq N$, denote by $(C,E_C,\mathcal{R}_C)$ the {\bf ranking problem} for coalition $C$. An allocation for $(C,E_C,\mathcal{R}_C)$ is a vector $x=(x_{i})_{i\in C}\in \mathbb{R}^{|C|}_{+}$ such that $\sum\nolimits_{i\in C}x_{i}=E_C$. A \textbf{(sharing) rule} $F$ is a function that associates with each problem $(C,E_C,\mathcal{R}_C)$ an allocation. Given a set of ranking problems $\{(C,E_C, \mathcal{R}_C)\}_{C\subseteq N, E_C \in \mathbb{R}_+}$, the \textbf{coalition formation problem induced by rule $\boldsymbol{F}$}, $\succsim^F$, is defined as in Section 3.\footnote{\cite{lucchetti2022coalition} consider a different approach to induce coalition formation problems using a ranking of agents (which is not exogenous, but rather deduced from a ranking over the different coalitions). They also focus on the study of core stable partitions. However, neither coalitional endowments nor sharing rules do appear in their model and, thus, our results cannot be applied therein.}

\medskip

\noindent We reformulate the family of \emph{interval rules} considered in \cite{dietzenbacher2023fair} for each ranking problem $(C, E_C, \mathcal{R}_C)$. Informally, given a ranking problem $(C, E_C, \mathcal{R}_C)$, each interval rule works as follows: first, a set of disjoint intervals is defined. Then, if the average endowment $\frac{E_C}{|C|}$ does not belong to any of the intervals, the endowment $E_C$ is equally split among the agents. Otherwise, the agents get the lower bound of the interval to which $\frac{E_C}{|C|}$ belongs. If there is endowment left, each agent is allocated up to the upper bound of that interval following the ranking $\mathcal{R}_C$. 
The formal definition is as follows:

\medskip

\noindent \textbf{Interval rules for ranking problems}: There exist disjoint intervals $(a_1, b_1), (a_2, b_2), \ldots$ with $a_1, a_2, \ldots \in \mathbb{R}_+$ and $b_1, b_2, \ldots \in \mathbb{R}_{+}\cup\{+\infty\}$ such that for each problem $(C,E_C,\mathcal{R}_C)$, 
$$I_{i}(C,E_C,\mathcal{R}_C)= \left\{ \begin{array}{lcc}
            a_k &   \mbox{if}  &  |C| a_k \leq E_C \leq ( |C|-\beta) a_k + \beta b_k; \\
            x &   \mbox{if}  & (|C|-\beta) a_k + \beta b_k \leq E_C \leq (|C| - \mathcal{R}_C (i))a_k + \mathcal{R}_C(i)b_k;\\
            b_k &   \mbox{if}  & ( |C| - \mathcal{R}_C(i))a_k + \mathcal{R}_C(i)b_k \leq E_C \leq  |C| b_k; \\
            \frac{E_C}{ |C|} & & \mbox{otherwise,} 

             \end{array}
   \right.$$
where $\beta=\mathcal{R}_C(i) - 1$ and $x = E -(|C| - \mathcal{R}_C(i))a_k-\beta b_k$.
 \medskip

\noindent As \cite{dietzenbacher2023fair} mention, the interval rule with $a_k = b_k = 0$ for each $k$ coincides with the Equal Division while the interval rule with $a_1 = 0$ and $b_1 = +\infty$ coincides with the Winner Takes All, both well-known rules.
\medskip

\noindent Theorem 1 in \cite{dietzenbacher2023fair} states that these are the only order-preserving\footnote{If ${\cal R}(i) < {\cal R}(j)$, then $F_i(C, E_C, {\cal R}_C) \geq F_j(C, E_C, {\cal R}_C)$.} rules that satisfy solidarity.\footnote{\cite{dietzenbacher2023fair} do not use the axiom of \textit{solidarity}, but the separate axioms of \textit{endowment monotonicity} and \textit{consistency}.} Consequently, our Theorem \ref{stable} yields the following.

\begin{corollary}\label{Bas}
The interval rules are the only order-preserving rules for ranking problems that induce non-circular coalition formation problems.
\end{corollary}

\noindent Corollary \ref{Bas} also implies that the interval rules guarantee stability in the induced coalition formation problems. Other interesting rules proposed by \cite{dietzenbacher2023fair} do not yield stability. The following example illustrates these features.
\medskip

\begin{example}
\noindent Let $N=\{1,2,3\}$ be such that $\mathcal{R}(i)=i$ for each $i \in N$. Assume the following coalitional endowments:\smallskip

\begin{table}[th]
\centering
\begin{tabular}{|c|c|c|c|c|c|}
\hline
$C$ & $\{12\}$ & $\{13\}$ & $\{23\}$ & $\{123\}$ & otherwise
\\ \hline
$E_C$ & $20$ &  $15$ & $14$ & $21$  & $0$\\ \hline
\end{tabular}%
\end{table}

\noindent Consider first an interval rule $I$ with $(a_1, b_1)=(2,9)$, $(a_2,b_2)=(9,10.5)$ and $(a_3, b_3)=(10.5, +\infty)$. This rule yields the following allocations:
\smallskip

\begin{table}[th!]
	\centering
		\begin{tabular}{|c|c|c|c|c|c|}
			\hline
			$C$ & $\{12\}$& $\{13\}$ & $\{23\}$ & $\{123\}$ & otherwise \\ \hline
			$I(C,E_C,\mathcal{R}_C)$ 
	& $(10.5, 9.5)$& $(9,6)$	& $(9,5)$ & $(9,9,3)$ & $(0)_{i\in C}$ \\ \hline
		\end{tabular}
	
\end{table}
\smallskip

\noindent As a consequence, the coalition formation problem induced by $I$, $\succsim^{I}$, is the following:

\begin{table}[th!]
\centering
\begin{tabular}{ccc}

 $\mathbf{1}$ & $\mathbf{2}$ & $\mathbf{3}$   \\ \hline
 $12$ & $12$ & $13$  \\ 
 $13 \sim123$ & $23 \sim123$ & $23$ \\ 
 $1$ & $2$ & $123$ \\  
  &  & $3$ \\ \hline
\end{tabular}%

\end{table}

\noindent Observe that this coalition formation problem satisfies weak pairwise alignment and has no rings. Then, it is a non-circular coalition formation problem. In particular, partition $\{\{12\},\{3\}\}$ is stable.

\medskip

\noindent We finally consider a class of rules for ranking problems based on the family of proportional rules defined in \cite{dietzenbacher2023fair}.\\

\noindent \textbf{Proportional rules for ranking problems}: Each of the rules in this family defines a set of parameters $\lambda^{1}, \lambda^{2}, \ldots \lambda^{|N|} \in \mathbb{R}_{+}$ such that $\lambda^{1}>0$ and $\lambda^{k}\ge\lambda^{k+1}$ for each $k\in \{1, \ldots, |N|-1\}$, and assigns an allocation to the agent in position $j$ in the ranking proportionally to $\lambda^j$. Formally, for each problem $(C,E_C,\mathcal{R}_C)$, $$PS_{i}(C,E_C,\mathcal{R}_C)=\frac{\lambda^{\mathcal{R}_{C}(i)}}{\sum_{j\in C} \lambda^{\mathcal{R}_{C}(j)}} \cdot E_{C}.$$
\noindent There exist rules within this family that do not satisfy solidarity and do not guarantee stability. An instance is the rule obtained when $|N| = 3$, $\lambda^{1}=3$ and $\lambda^{2} = \lambda^{3}=1$, which yields the following allocations in this problem:
\smallskip

\begin{table}[th!]
\centering
\begin{tabular}{|c|c|c|c|c|c|}
\hline
$C$ &  $\{12\}$ & $\{13\}$ &  $\{23\}$ & $\{123\}$ & otherwise\\ \hline
$PS(C,E_C,\mathcal{R}_C)$ & $(15,5)$ &  $(11.25, 3.75)$ & $(10.5, 3.5)$ & $(12.6,4.2,4.2)$ & $(0)_{i\in C}$ \\ \hline
\end{tabular}
\end{table}
\smallskip

\noindent As a consequence, the coalition formation problem induced by $PS$, $\succsim^{PS}$, is the following:

\begin{table}[H]
\centering
\begin{tabular}{ccc}
$\mathbf{1}$ & $\mathbf{2}$ & $\mathbf{3}$ \\ \hline
$12$ & $23$ & $123$ \\ 
$123$ & $12$ & $13$ \\ 
$13$ & $123$ & $23$ \\
$1$ & $2$ & $3$ \\  \hline
\end{tabular}
\end{table}

\noindent Note that this problem is not a non-circular coalition formation problem. In particular, $(\{12\}$, $\{23\}$, $\{13\})$ is a ring and, as $23\succ_2 123$, while $123\succ_3 23$, \textit{weak pairwise alignment} is also violated. Moreover, it can easily be shown that this problem has an empty core.
\end{example}

\subsection*{Sharing problems with permissible coalitions}

In some resource allocation situations not all coalitions are feasible and, consequently, our results cannot be directly applied. Instances are matching problems \citep[e.g.,][]{demange1985strategy, roth_sotomayor_1990}, network games \citep[e.g.,][]{jackson2003strategic, jackson2005allocation}, river sharing problems \citep[e.g.,][]{ambec2002sharing, alcalde2015sharing} or legislative bargaining problems \citep[e.g.,][]{breton2008gamson, puy2013stable}. However, we propose a way to partially circumvent this issue. Suppose that the resource allocation situation involves just a collection of permissible coalitions ${\cal K} \subset 2^N$ and, therefore, sharing problems are only defined for these coalitions.\footnote{To ensure the existence of partitions, we assume that singletons are permissible (i.e., for each $i \in N$, $\{i\} \in {\cal K}$).} We propose to define, for each non-permissible coalition, an auxiliary sharing problem with a zero endowment. We then define an enlarged model that encompasses the sharing problems for permissible coalitions as well as the auxiliary sharing problems for non-permissible coalitions. Then, Theorem \ref{stable} and Corollary \ref{stable1} can be applied to this enlarged model. Thus, the induced coalition formation problem has a non-empty core if the sharing rule satisfies solidarity. If so, note that there is a stable partition formed by coalitions of ${\cal K}$, given that only those may have a positive endowment. Clearly, such a partition is also stable in the original model. To summarize, given a resource allocation situation where not all coalitions are permissible, we can always obtain a stable partition by constructing such an enlarged model, and then applying a sharing rule that satisfies solidarity therein. Observe that solidarity refers to the sharing rule in the enlarged model, as our results require that all coalitions are feasible. However, the application of a sharing rule satisfying solidarity in the original model does not guarantee stability. We illustrate this discussion with the following example.

\begin{example} \noindent Let $N = \{1,2, 3\}$ and $\mathcal{K}=\{\{12\},\{13\},\{23\},\{1\},\{2\},\{3\}\}$ be the set of permissible coalitions. Assume the following coalitional endowments:

\begin{table}[H]
\centering
\begin{tabular}{|c|c|c|c|c|}
\hline
$C$ & $\{12\}$ & $\{13\}$ & $\{23\}$ & $\{i\}$ 
\\ \hline
$E_C$ & $9$ & $9$ & $9$ & $0$  \\ \hline
\end{tabular}%
\end{table}

\noindent Now, consider the following sharing rule $F$: for each $(C, E_C) \in {\cal P}$, with $C \in {\cal K}$, and each $i \in C$, \smallskip

$$F_i(C, E_C)=
\left\{
\begin{array}{cc}
\frac{2}{3}E_C & \mbox{ if } [i=1 \mbox{ and } C=\{12\}]  \mbox{ or } [i=2  \mbox{ and } C=\{23\}],  \mbox{ or } [i=3  \mbox{ and } C=\{13\}],\\
\frac{1}{3}E_C & \mbox{ if } [i=2  \mbox{ and } C=\{12\}] \mbox{ or } [i=3  \mbox{ and } C=\{23\}],  \mbox{ or } [i=1  \mbox{ and } C=\{13\}],\\
E_C & \mbox{ if } C=\{i\}.\\
\end{array}
\right.
$$

\noindent Note that this rule gives priority to one agent within each pair, but whoever has the priority depends on the coalition. It can be checked that this rule satisfies \textit{solidarity}. Then, $F$ yields the following allocations:
\smallskip

\begin{table}[h!]
	\centering
		\begin{tabular}{|c|c|c|c|c|}
			\hline
			$C$ &  $\{12\}$&  $\{13\}$ & $\{23\}$ & $\{i\}$ \\ 
			\hline
		$F(C, E_C)$ &$(6,3)$	&$(3,6)$ & $(6,3)$ & $(0)$ \\ \hline
		\end{tabular}
	
\end{table}


\noindent As a consequence, the coalition formation problem induced by $F$, $\succsim^{F}$, is the following:
\newpage
\smallskip

\begin{table}[h!]
\centering
\begin{tabular}{ccc}

 $\mathbf{1}$ & $\mathbf{2}$ & $\mathbf{3}$  \\ \hline
 $12$ & $23$ & $13$  \\ 
$13$ & $12$ & $23$  \\ 
$1$ & $2$ & $3$  \\ 
\end{tabular}%

\end{table}

\noindent Observe that $(\{12\},\{23\},\{13\})$ is a ring and that this coalition formation problem has an empty core.

\medskip

\noindent Consider now the enlarged model in which the grand coalition gets a zero endowment: 

\begin{table}[H]
\centering
\begin{tabular}{|c|c|c|c|c|c|}
\hline
$C$ & $\{12\}$ & $\{13\}$ & $\{23\}$ & $\{i\}$ & $\{123\}$
\\ \hline
$E_C$ & $9$ & $9$ & $9$ & $0$ & $0$  \\ \hline
\end{tabular}%
\end{table}

\noindent We know, by Corollary 1, that the application to this enlarged model of any sharing rule that satisfies solidarity guarantees the non-emptiness of the core of the induced coalition formation problem. More specifically, we can guarantee that there is a stable partition in which the grand coalition does not take part. Consider, for instance, the  uniform sharing rule, US, that allocates the endowment of each coalition equally among its members (and, thus, it obviously satisfies solidarity).\footnote{Formally, for each $(C, E_C) \in {\cal P}$ and each $i \in C$, $US_i(C, E_C) = \frac{E_C}{|C|}$.} Then, $US$ yields the following allocations:
\smallskip

\begin{table}[h!]
	\centering
		\begin{tabular}{|c|c|c|c|c|c|}
			\hline
			$C$ &  $\{12\}$&  $\{13\}$ & $\{23\}$ & $\{i\}$ & $\{123\}$ \\ 
			\hline
		$US(C, E_C)$ &$(4.5,4.5)$ & $(4.5,4.5)$ & $(4.5,4.5)$ & $(0)$ & $(0,0,0)$ \\ \hline
		\end{tabular}
	
\end{table}

\noindent As a consequence, the coalition formation problem induced by $US$, $\succsim^{US}$, is the following:


\begin{table}[h!]
\centering
\begin{tabular}{ccc}

$\mathbf{1}$ & $\mathbf{2}$ & $\mathbf{3}$  \\ \hline
$12$ $\sim$ $13$ & $12$ $\sim$ $23$ & $13$ $\sim$ $23$  \\ 
$1$ $\sim$ $123$ & $2$ $\sim$ $123$ & $3$ $\sim$ $123$  \\ 
\end{tabular}%

\end{table}
\smallskip

\noindent Observe that this problem has three stable partitions: $\{\{12\},\{3\}\}$, $\{\{13\},\{2\}\}$ and $\{\{23\},\{1\}\}$, all of them formed by permissible coalitions.

\end{example}

\section{Concluding Remarks}

We have studied in this paper coalition formation problems in a context in which coalitions have to share collective resources. We have characterized the sharing rules that induce non-circular coalition formation problems as those satisfying a natural axiom formalizing the principle of solidarity. This implies that such a solidarity axiom guarantees the existence of (core) stable partitions in the induced coalition formation problem. Our result can be applied to canonical problems of resource allocation long studied such as bargaining, rationing, or surplus sharing problems as well as to other problems recently considered such as ranking problems. 

\medskip

\noindent Although our benchmark model requires that all coalitions are feasible, we have also seen that we can partially apply the results to resource allocation situations where not all coalitions are feasible. In contrast, a similar argument cannot be applied to situations where agents are equipped with individual endowments such as revenue sharing in hierarchies \citep[e.g.,][]{hougaard2017sharing, Patrick}. Exploring whether the connection between solidarity (in the resource allocation problem) and stability (in the corresponding coalition formation problem) extends to these cases is left for further research.



\medskip

\noindent The class of non-circular coalition formation problems is not the only one in which stability is guaranteed. As mentioned, various other properties have been introduced to guarantee the existence of stable partitions in hedonic games. Most of them, like the \emph{top-coalition property} or the \emph{ordinal balancedness property} \cite[]{bogomolnaia2002stability} are sufficient conditions to guarantee the non-emptiness of the core. The somewhat related notion of \emph{pivotal balancedness}, introduced by \cite{iehle2007core}, is both necessary and sufficient for the existence of a core stable partition. Nevertheless, all these properties guarantee the non-emptiness of the core even though they may allow for the presence of rings.\footnote{See the example at page 213 in \cite{bogomolnaia2002stability} for a coalition formation problem satisfying the \emph{ordinal balancedness property} and having a ring but also a non-empty core. The same example is valid for the \emph{pivotal balancedness property}.} A natural research question would be to identify properties of sharing rules that induce coalition formation problems satisfying, for instance, ordinal balancedness or pivotal balancedness. However, this is a challenging question as these properties allow for the existence of rings in agents' preferences. And, as it is known, the non-emptiness of the core when rings exist is contingent on various factors, including the number of coalitions that form the ring and their positions within the preferences of the agents involved \cite[see][]{bonifacio2022stable}. For instance, this last factor is relevant because the higher these positions are, \emph{ceteris paribus}, the more difficult is to have a non-empty core.\footnote{To see the importance of these positions in the non-emptiness of the core when rings exist, see, for instance, Examples 2 and 3 in \cite{gallo2018rationing}.} In our model, these positions depend significantly on the coalitional endowments. Note also that we give full flexibility to the values of these coalitional endowments. Therefore, if a particular sharing rule can generate a ring for some coalitional endowments, changing the endowments of all the coalitions outside the ring to $0$ would suffice to generate a coalition formation problem (induced by that rule) with all the coalitions of the ring occupying the first positions of agents' preferences. Therefore, it seems very difficult to formulate general properties on sharing rules that only generate rings compatible with non-empty core coalition formation problems. A hypothetical characterization of sharing rules that allow for rings but guarantee a non-empty core would be based on restricting the values for coalitional endowments. To summarize, the search for a characterization of the sharing rules that induce coalition formation problems satisfying any of the above conditions (that allow for rings but guarantee a non-empty core) is a daunting task beyond the scope of this paper. We have focused instead on the non-circular property, which conveys the appealing feature of excluding rings altogether (an aspect that can be naturally linked to sharing rules).

\medskip

\noindent Finally, another interesting question is the study of the computational complexity of coalition formation problems. \cite{ballester2004np} studies the complexity of coalition formation games and shows that the computation of stable partitions is NP-complete. The computational complexity of stable partitions in additive coalition formation problems has been studied by \cite{sung2010computational}. More recently, \cite{gairing2019computing} focus on symmetric additively separable coalition formation problems. The study of the computational complexity of non-circular coalition formation problems is left for further research.

\section*{Appendix: Proof of Theorem \ref{stable}}

\noindent We start by showing that if $F$ satisfies solidarity, then $\succsim^F$ always satisfies \emph{weak pairwise alignment}.

\begin{lemma}
\label{wpa}
If $F$ satisfies solidarity, then $\succsim^F$ satisfies \textit{ weak pairwise alignment} for each $\{E_C\}_{C \subseteq N}$.
\end{lemma}

\begin{proof}
Let $F$ be a sharing rule that satisfies \textit{solidarity}. Then, by Lemma \ref{rm+c}, $F$ satisfies \textit{endowment monotonicity} and \textit{consistency}. Let $C,C'\subseteq N$, $E_C, E_{C'} \in \mathbb{R}_+$, and $i,j\in C\cap C'$.

\medskip

\noindent If $C\subset C'$ or $C'\subset C$, then by \textit{solidarity}, either [$F_{k}(C,E_C) \leq F_{k}(C',E_{C'})$ for each $k\in\{i,j\}$] or [$F_{k}(C,E_C) \geq F_{k}(C',E_{C'})$ for each $k\in\{i,j\}$]. Therefore, agents $i$ and $j$ do not rank $C$ and $C'$ in opposite ways.

\medskip

\noindent Otherwise, $C\not\subset C'$ and $C'\not\subset C$. Then, let $x_k=F_k(C,E_{C})$ for each $k \in C$, and $x'_{k'}=F_{k'}(C',E_{C'})$ for each $k' \in C'$. Consider the sharing problems $(\{i,j\}, x_{i}+x_{j}), (\{i,j\}, x'_i + x'_j)$. By \textit{consistency},
\begin{equation*}
(x_i, x_j)=F(\{i,j\}, x_{i}+x_{j})\text{ and } (x'_i, x'_j) = F(\{i,j\}, x'_{i}+x'_{j}).
\end{equation*}%

\noindent Assume, without loss of generality, that $x_{i} + x_{j} \geq x'_{i}+x'_{j}$. Then, by \textit{endowment monotonicity}, for each $k\in\{i,j\}$, $x_{k}=F_{k}(\{i,j\}, x_{i}+x_{j}) \geq F_{k}(\{i,j\},x'_{i}+x'_{j})=x'_{k}$. Therefore, for each $k\in\{i,j\}$, $F_{k}(C,E_{C})\geq F_{k}(C',E_{C'})$. 
Consequently, agents $i$ and $j$ do not rank $C$ and $C'$ in opposite ways.

\medskip

\noindent Hence, $\succsim^{F}$ satisfies \textit{weak pairwise alignment}, as desired. 
\end{proof}

\noindent We now show that if $F$ satisfies solidarity, then $\succsim^F$ has no rings. To prove this, we first introduce the following auxiliary lemma. 

\begin{lemma}
\label{auxiliary} Let $F$ be a sharing rule that satisfies \textit{solidarity}. Let $C \subseteq N$ and $\{C_1, \ldots, C_m\}$ be a set of coalitions such that $\cup_{k=1}^m C_k = C$. Then, there is $C_l \in \{C_1, \ldots, C_m\}$ such that for each $E_{C_l} \in \mathbb{R}_+$, there exists $E_C \in \mathbb{R}_+$ for which $F_i(C_l, E_{C_l}) = F_i(C, E_C)$ for each $i \in C_l$.
\end{lemma}

\begin{proof}
Let $F$ be a sharing rule that satisfies \textit{solidarity}. Then, by Lemma \ref{rm+c}, $F$ satisfies \textit{endowment continuity} and \textit{consistency}. Let $C \subseteq N$ and $\{C_1, \ldots, C_m\}$ be such that $\cup_{k=1}^m C_k = C$.

\medskip

\noindent We first prove that there exists $j \in C$ such that $\lim_{E_C \rightarrow \infty} F_j(C, E_C) = \infty$. Suppose otherwise. Then, for each $j \in C$, there exists $M_j \in \mathbb{R}_+$ such that $F_j(C, E) < M_j$ for each $E\in \mathbb{R}_+$ arbitrarily large. Denote $M \equiv \max_{j \in C} M_j$. Then, $F_j(C, |C| M) < M$ for each $j \in C$ and, therefore, $\sum_{j\in C} F_j(C, |C| M) < |C| M$, which contradicts the definition of an allocation.

\medskip

\noindent Now, as $\cup_{k=1}^m C_k = C$, there is $C_l \in \{C_1, \ldots, C_m\}$ such that $j \in C_l$. We then construct $\alpha : \mathbb{R}_+ \rightarrow \mathbb{R}_+$ such that for each $\xi \in \mathbb{R}_+$, $\alpha (\xi) = \sum_{i \in C_l} F_i(C, \xi)$. Then, $\alpha(0) =\sum_{i \in C_l} F_i(C, 0) = 0$ and $\lim_{\xi \rightarrow \infty} \alpha(\xi)=\lim_{\xi \rightarrow \infty} \sum_{i \in C_l} F_i(C, \xi) = \infty$.
As $F$ satisfies \textit{endowment continuity}, $\alpha$ is continuous. Then, for each $E_{C_l} \in \mathbb{R}_+$, there exists $E_C \in \mathbb{R}_+$ for which $\alpha(E_C) = E_{C_l}$. By \textit{consistency}, for each $i \in C_l$, $F_i(C_l, E_{C_l}) = F_i(C, E_C)$, as desired.
\end{proof}


\begin{lemma}
\label{noring}
If $F$ satisfies solidarity, then for each $\{E_C\}_{C \subseteq N}$, $\succsim^F$ has no rings.
\end{lemma}

\begin{proof}
Let $F$ be a sharing rule that satisfies \textit{solidarity}. Suppose by contradiction that there exists $\{E_C\}_{C \subseteq N}$ such that $\succsim^{F}$ has a ring $(C_0, C_1 \ldots, C_{l-1})$. Then, for each subscript $k=0, \ldots, l-1$ (modulo $l$), there is at least one agent, say agent $j_{k+1} \in C_{k+1} \cap C_{k}$, such that $C_{k+1} \succ_{j_{k+1}}^{F} C_{k}$. 

\medskip

\noindent Let $\hat{C} = \cup_{k=0}^{l-1} C_k$. By Lemma \ref{auxiliary}, there is $C_k \in \{C_0, C_1, \ldots, C_{l-1}\}$ and $\hat{E}_{\hat{C}} \in \mathbb{R}_+$ such that for each $i \in C_k$, $F_i(C_k, E_{C_k}) = F_i(\hat{C}, \hat{E}_{\hat{C}})$. Assume without loss of generality that $C_k = C_0$.

\medskip

\noindent Consider now $\{E'_C\}_{C \subseteq N}$ such that $E'_C=E_C$ for each $C \in 2^N \setminus \{\hat{C}\}$ and $E'_{\hat{C}} = \hat{E}_{\hat{C}}$. We denote by $\succsim^{F'}$ the coalition formation problem when $F$ is applied and the endowments are $\{E'_C\}_{C \subseteq N}$.

\medskip

\noindent By construction, for each $i\in C_0$, $F_i(\hat{C}, E'_{\hat{C}}) = F_i(C_0, E'_{C_0})$ and, therefore, $C_0 \sim^{F'}_{i} \hat{C}$. In particular, $C_0 \sim^{F'}_{j_0} \hat{C}$ and $C_0 \sim^{F'}_{j_1} \hat{C}$ (possibly $j_0=j_1$). 
Similarly, for each $i'\in C_k \cap C_{k+1}$, $k = 0, 1, \ldots, l-1$ (modulo $l$), $F_{i'}(C_k, E'_{C_k}) = F_{i'}(C_k, E_{C_k})$ and $F_{i'}(C_{k+1}, E'_{C_{k+1}}) = F_{i'}(C_{k+1}, E_{C_{k+1}})$, which implies that $C_{k+1} \succ_{i'}^{F'} C_k \Leftrightarrow C_{k+1} \succ_{i'}^{F} C_k$ and $C_{k+1} \sim_{i'}^{F'} C_k \Leftrightarrow C_{k+1} \sim_{i'}^{F} C_k$. In particular, for each $k=0, 1, \ldots, l-1$ (modulo $l$), $C_{k+1} \succ_{j_{k+1}}^{F'} C_{k}$. Then, by transitivity, $\hat{C} \succ^{F'}_{j_{0}} C_{l-1}$ and $C_1 \succ^{F'}_{j_1} \hat{C}$. As $F$ satisfies \textit{solidarity}, it follows
by Lemma \ref{wpa} that $\succsim^{F'}$ satisfies \textit{weak pairwise alignment} and, therefore, $C_1 \succsim_{j_2}^{F'} \hat{C}$. 
 As $C_2 \succ_{j_2}^{F'} C_1$, it follows by transitivity that $C_2 \succ_{j_2}^{F'} \hat{C}$. 
Similarly, we have that $C_k \succ_{j_k}^{F'} \hat{C}$ for each  $k \in \{3, \ldots, l-1\}$. Then, we have deduced that $\hat{C} \succ_{j_0}^{F'} C_{l-1}$ and $C_{l-1} \succ_{j_{l-1}}^{F'} \hat{C}$. If $j_{l-1} = j_{0}$, this contradicts transitivity. Otherwise, this implies that $\succsim^{F'}$ does not satisfy \textit{weak pairwise alignment},  which contradicts Lemma \ref{wpa}.

\end{proof}

\noindent Lemmas \ref{wpa} and \ref{noring} prove one implication of Theorem \ref{stable}, while the other is proven by the following lemma.

\begin{lemma}
\label{neces cond}
If $F$ does not satisfy solidarity, then there is $\{E_C\}_{C \subseteq N}$ such that $\succsim^{F}$ does not satisfy\textit{ weak pairwise alignment}.
\end{lemma}

\begin{proof}
Let $F$ be a sharing rule that does not satisfy \textit{solidarity}. Then, there exist $C, C' \subseteq N$, with $C' \subset C$, $i, j \in C'$ and $E_C, E_{C'} \in \mathbb{R}_+$ such that $F_i(C, E_C) > F_i(C', E_{C'})$ and $F_j(C, E_C) < F_j(C', E_{C'})$. Then, for $\{E_C\}_{C \subseteq N}$, we have that $C \succ_{i}^{F} C'$ and $C'\succ_j^{F} C$. Hence, $\succsim^{F}$ does not satisfy\textit{ weak pairwise alignment}.
\end{proof}

\bibliographystyle{ecta}
\bibliography{CoalitionFormation}

\end{document}